\documentclass[12pt,draftclsnofoot,journal,onecolumn]{IEEEtran}

\usepackage{amsmath}
\usepackage{amssymb}
\usepackage{graphicx}
\usepackage{bm}
\usepackage{paralist}
\usepackage{subfigure}
\usepackage{color}

\newcommand{\bx}{\bm{x}}
\newcommand{\bc}{\bm{c}}

\newcommand{\bz}{\ensuremath{{\bm z}}}

\newcommand{\br}{\ensuremath{{\bm r}}}
\newcommand{\by}{\ensuremath{{\bm y}}}

\newcommand{\brho}{\ensuremath{{\bm \rho}}}

\newcommand{\set}[1]{\ensuremath{\mathcal #1}}

\newcommand{\separator}{
  \begin{center}
    \rule{\columnwidth}{0.3mm}
  \end{center}
}

\def\eg{{\it e.g.}}
\def\ie{{\it i.e.}}

\newtheorem{theorem}{Theorem}[section]
\newtheorem{corollary}{Corollary}[section]

\newtheorem{lemma}{Lemma}[section]
\newtheorem{definition}{Definition}[section]

\newcommand{\beq}{\begin{eqnarray*}}
\newcommand{\eeq}{\end{eqnarray*}}
\newcommand{\beqn}{\begin{eqnarray}}
\newcommand{\eeqn}{\end{eqnarray}}
\newcommand{\bemn}{\begin{multiline}}
\newcommand{\eemn}{\end{multiline}}

\newcommand{\sqeq}{\addtolength{\thinmuskip}{-4mu}
\addtolength{\medmuskip}{-4mu}\addtolength{\thickmuskip}{-4mu}}

\newcommand{\unsqeq}{\addtolength{\thinmuskip}{+4mu}
\addtolength{\medmuskip}{+4mu}\addtolength{\thickmuskip}{+4mu}}

\def\cI{{\cal I}}

\def\R{\mathbb{R}}

\newcommand{\bgamma}{\boldsymbol{\gamma}}
\newcommand{\blambda}{\boldsymbol{\lambda}}
\newcommand{\Q}{\boldsymbol{Q}}
\newcommand{\ratech}{\gamma}
\newcommand{\rateon}{\gamma^{(0\to 1)}}
\newcommand{\rateoff}{\gamma^{(1\to 0)}}
\newcommand{\bsigma}{\boldsymbol{\sigma}}
\newcommand{\bone}{\mathbf{1}}
\newcommand{\bD}{\boldsymbol{D}}
\newcommand{\cH}{\boldsymbol{C}}
\newcommand{\boR}{\boldsymbol{R}}
\newcommand{\bor}{\boldsymbol{r}}
\newcommand{\boS}{\boldsymbol{S}}
\newcommand{\bLambda}{\boldsymbol{\Lambda}}
\newcommand{\sch}{\mathcal{H}}
\newcommand{\chs}{h}

\newcommand{\bof}{\boldsymbol{f}}
\newcommand{\bog}{\boldsymbol{g}}
\newcommand{\boc}{\boldsymbol{c}}

\begin{document}

%
 \title{CSMA over Time-varying Channels: \\
   Optimality, Uniqueness and \\ Limited Backoff Rate}


\author{Se-Young Yun, Jinwoo Shin, and Yung Yi \thanks{This work was
     supported by the Center for Integrated Smart Sensors funded by
     the Ministry of Education, Science and Technology as Global
     Frontier Project (CISS-2012M3A6A6054195).}\thanks{S. Yun is with
     School of Electrical Engineering, KTH, Sweden (e-mail:
     seyoung@kth.se) and this work is done while he was with Department of Electrical
     Engineering, KAIST, Korea. J. Shin and Y. Yi are with Department of Electrical
     Engineering, KAIST, Korea (e-mail: mijirim@gmail.com and yiyung@kaist.edu).  }}

\maketitle

\thispagestyle{empty}

\begin{abstract}
  Recent studies on MAC scheduling have shown that carrier sense
  multiple access (CSMA) algorithms can be throughput optimal for
  arbitrary wireless network topology.  However, these results are
  highly sensitive to the underlying assumption on `static' or `fixed'
  system conditions. For example, if channel conditions are
  time-varying, it is unclear how each node can adjust its CSMA
  parameters, so-called backoff and channel holding times, using its
  local channel information for the desired high performance. In this
  paper, we study `channel-aware' CSMA (A-CSMA) algorithms in
  time-varying channels, where they adjust their parameters as some
  function of the current channel capacity. First, we show that the
  achievable rate region of A-CSMA equals to the maximum rate region
  if and only if the function is exponential. Furthermore, given an
  exponential function in A-CSMA, we design updating rules for their
  parameters, which achieve throughput optimality for an
  arbitrary wireless network topology.  They are the first CSMA
  algorithms in the literature which are proved to be throughput
  optimal under time-varying channels.  Moreover, we
  also consider the case when back-off rates of A-CSMA are highly
  restricted compared to the speed of channel variations, and
  characterize the throughput performance of A-CSMA in terms of the
  underlying wireless network topology. Our results not only guide a
  high-performance design on MAC scheduling under highly time-varying
  scenarios, but also provide new insights on the performance of CSMA
  algorithms in relation to their backoff rates and the network
  topology.
\end{abstract}



%

\section{Introduction}
\subsection{Motivation}
How to access the shared medium is a crucial issue in achieving high
performance in many applications, \eg, wireless networks.  In spite of a
surge of research papers in this area, it's the year 1992 that the
seminal work by Tassiulas and Ephremides proposed a throughput optimal
medium access algorithm, referred to as Max-Weight \cite{TE92}. Since
then, a huge array of subsequent research has been made to develop
distributed medium access algorithms with high performance guarantees and
low complexity. However, in many cases the tradeoff between complexity
and efficiency has been observed, or even throughput optimal algorithms
with polynomial complexity have turned out to require heavy message
passing, which becomes a major hurdle to becoming practical medium
access schemes, \eg, see \cite{JWbook10,YC11-survey} for surveys.

Recently, there has been exciting progresses that even fully
distributed medium access algorithms based on CSMA (Carrier Sense
Multiple Access) with no or very little message passing can achieve
optimality in both throughput and utility, \eg, see
\cite{JW10DC,IEEEpaper:Liu_Yi_Proutiere_Chiang_Poor_2009,SRS09,NTS10QCSMA}. 
The main intuition underlying these results is that nodes dynamically
adjust their CSMA parameters, {\em backoff} and {\em channel holding}
times, using local information such as queue-length so that they solve
a certain network-wide optimization problem for the desired high
performance. There is a survey paper which dealing with recent results
on the CSMA algorithms~\cite{YY12OC}.



However, the recent CSMA algorithms crucially rely on 
the assumption of static channel conditions. 
It is far from being clear how they perform for time-varying channels,
which frequently occurs in practice. 
Note that it has already been shown that the Max-Weight is throughput optimal
for time-varying channels \cite{Ta97SP} and joint scheduling and
congestion control algorithms based on the optimization decomposition,
\eg, \cite{GNT06RA}, are utility optimal by selecting the schedules
over time, both of which essentially track the channel conditions
quickly. However, a similar channel adaptation for CSMA algorithms may not
be feasible for the following two reasons. First, each node in a
network only knows its local channel information, and cannot track
channel conditions of other
nodes. 
Second, there exists a non-trivial coupling between CSMA's performance
under time-varying channels and the speed of channel variations. A CSMA
schedule at some instant may not have enough time to be close to the
desired `stationary' distribution before the channel changes.
In this paper, we formalize and quantify this coupling, and
  study when and how CSMA algorithms perform depending on the network
  topologies and the speed of channel variations.

\subsection{Our Contribution}

In this paper, we model time-varying channels by a Markov process, and
study `channel-aware' CSMA (A-CSMA) algorithms where each link adjusts
its CSMA parameters, backoff and channel holding times, as some
function of its (local) channel capacity. In what follows, we first
summarize our main contributions and then describe more details.

\smallskip
\noindent {\bf \em C1 -- Achievable rate region of A-CSMA}.  We show
that the achievable rate region of A-CSMA is maximized if and only if the
function is exponential.  In particular, we prove that A-CSMA can
achieve an arbitrary large fraction of the capacity region for
exponential functions (see Theorem~\ref{thm3}), which turns out to be
{\em impossible} for non-exponential functions (see
Theorem~\ref{thm:nonlinear}).  

\smallskip
\noindent {\bf \em C2 -- Dynamic throughput optimal A-CSMA}.  
We develop two types of throughput optimal A-CSMA algorithms, where
links dynamically update their CSMA
parameters based on both (a) the exponential function of the channel
capacity in {\bf \em C1} and (b) the empirical local load or the local queue length,
without knowledge of the speed of channel variation and the arrival
statistics (such as its mean) in advance (see Theorem
\ref{thm:timevaryingcsma} and \ref{thm:queuecsma}).

\smallskip
\noindent {\bf \em C3 -- Achievable rate region of A-CSMA with
  limited backoff rates.} 
We provide a lower bound for the achievable rate region of A-CSMA when 
their backoff rates are highly limited compared to the speed of channel variations (see Theorem \ref{thm:achi-rate-regi}).
Our bound depends on a combinatorial property of the underlying
interference graph (\ie, its chromatic number), and is independent of
backoff rates or the speed of channel variations. Moreover, it is
noteworthy that the achievable rate region of A-CSMA includes the
achievable rate region of channel-unaware CSMA (U-CSMA) for any
limited backoff rate (see Corollary~\ref{cor:UACSMA}).


\smallskip 
A typical necessary step to analyze and design a CSMA
algorithm of high performance in static channels is to characterize
the stationary distribution of the Markov chain induced by it
\cite{JW10DC,IEEEpaper:Liu_Yi_Proutiere_Chiang_Poor_2009,SRS09,NTS10QCSMA}.
However, this task is much harder for A-CSMA in time-varying channels,
since the Markov chain induced by A-CSMA is {\em non-reversible} (see Theorem~\ref{thm:nonrev}), \ie,
it is unlikely that its stationary distribution has a `clean' formula
to analyze, being in sharp contrast to the CSMA analysis for static
channels. To overcome this technical issue, we first show that the
stationary distribution approximates to a 
of product-form distribution when backoff rates are large
enough. Then, for {\bf \em C1}, we study the product-form to guarantee
high throughput of A-CSMA, where the exponential functions are found.
The main novelty lies in establishing the approximation scheme, using
the {\em Markov chain tree theorem} \cite{AT89PMCT}, which requires
counting the weights of arborescences induced by the non-reversible
Markov process to understand its stationary distribution.

For {\bf \em C2}, we combine {\bf \em C1}
with existing techniques: our first and second throughput optimal algorithms 
are `rate-based' and `queue-based' ones originally studied
in static channels by Jiang et al.\ (cf.\ \cite{JW10DC,JSSW10DRA}) and
Rajagopalan et al.\ (cf.\ \cite{SRS09,Shin12}), respectively. To extend
these results to time-varying channels, our specific choice of holding
times as exponential functions of the channel
capacity plays a key role in establishing the desired throughput optimal
performance. To our best knowledge, they are 
the first CSMA algorithms in the literature which are proved to be throughput optimal 
under general Markovian time-varying channel models.

{\bf \em C3} is motivated by observing
that a CSMA algorithm in fast time-varying channels
inevitably has to be of high backoff rates for the desired throughput performance, \ie,
high backoff rates are needed for tracking time-varying channel conditions fast enough.
However, backoff rates are bounded in practice, which may cause degradation in the CSMA's
performance. We note that CSMA algorithms with limited backoff or holding rates 
have been little analyzed in the literature, despite of their practical importance.\footnote{Even in static channels,
restrictions on backoff or holding rates may degrade the throughput or delay performances of CSMA algorithms.}
{\bf \em C3} provides a lower bound for A-CSMA throughputs regardless of restrictions on
their backoff rates or sensing frequencies. For example, if the interference graph
is bipartite (\ie., its chromatic number is two), our bound implies that A-CSMA is guaranteed to have 
at least 50\%-throughput even with arbitrary small backoff rates. 
Furthermore, one can design a dynamic high-throughput A-CSMA algorithm with limited backoff
rates using {\bf \em C3} (similarly as {\bf \em C1} is used for {\bf \em C2}), 
but in the current paper we do not present further details due to
space limitation.


\subsection{Related Work}

The research on throughput optimal CSMA has been initiated independently by 
Jiang et al.\ (cf.\ \cite{JW10DC,JSSW10DRA}) and Rajagopalan et al.\
(cf.\ \cite{SRS09,Shin12}), where both consider the continuous time
and collision free setting. Under exponential
distributions on backoff and holding times, the system is modeled by a
continuous time Markov chain, where the backoff rate or channel holding
time at each link is adaptively controlled to the local (virtual or actual) queue
lengths. 
Jiang et al.\ proved that the long-term link throughputs are the
solution of an utility maximization problem assuming the infinite backlogged
data. Rajagopalan et al.\ showed that
if the CSMA parameters are changing very slowly with respect to the queue length changes, 
 the mixing time is much faster than the queue length changes so that
the realized link schedules can provably emulate Max-Weight very well. 
Although their key intuitions are apparently different, analytic techniques
are quite similar, \ie, both require to understand the long-term behavior (\ie\
stationarity) of the Markov chains formed by CSMA.

These throughput optimality results motivate further research on design and analysis of CSMA algorithms. 
The work by Liu et al.\
\cite{IEEEpaper:Liu_Yi_Proutiere_Chiang_Poor_2009} follows the approach in \cite{JW10DC}
and proves the utility optimality using a stochastic approximation
technique, which has been extended to the multi-channel, multi-radio
case with a simpler proof in \cite{PY10RA}. The throughput optimality of MIMO networks under SINR model is also shown in \cite{QZ10CDM}.
As opposed to the continuous-time setting that carrier sensing is perfect and
instantaneous (and hence no collision occurs), more practical discrete
time settings
that carrier sensing is imperfect or delayed (and hence collisions occur) 
have been also studied. 
The throughput optimality of CSMA algorithms in discrete time settings with collisions is
established in \cite{JW11}, \cite{shah2011medium} and \cite{KNSV11ICS}, where
the authors in \cite{KNSV11ICS} consider imperfect sensing information.
In \cite{IEEEpaper:Liu_Yi_Proutiere_Chiang_Poor_2009}, the authors
studied the impact of collisions and the tradeoff between short-term
fairness and efficiency. The authors in \cite{NTS10QCSMA} considered a
synchronous system consisting of the control phase, which eliminates the
chances of data collisions via a simple message passing, and the data
phase, which actually enables data transmission based on the
discrete-time Glauber dynamics. 
There also exist several efforts on improving or analyzing delay
performance \cite{JLNSW11FM,Sinclair07,LE12FC,Marbach11TO,Shah10DO,CD12GM}.

To the best of our knowledge, CSMA under time-varying channels has been
studied only in \cite{LE12FC} for only complete interference graphs,
when the arbitrary backoff rate is allowed, and more seriously, under
the time-scale separation assumption, which does not often hold in
practice and extremely simplifies the analysis (no mixing time related
details are needed).



\section{Model and Preliminaries}
\label{sec:model}

\subsection{Network Model}

We consider a network consisting of a collection of $n$ queues (or
links) $\{1,\dots,n\}$ and time is indexed by $t\in \mathbb{R}_+$. Let
$Q_i(t) \in \mathbb{R}_+$ denote the amount of work in queue $i$ at time
$t$ and let $\Q(t) = [Q_i(t)]_{1\leq i\leq n}$.  The system starts
empty, \ie, $Q_i(0)=0$. We assume work arrives at each queue $i$ as per an
exogenous ergodic stationary process with rate
$\lambda_i>0$, \ie,
$$E\left[A_i(t,t+1)~|~A_i(0,t)\right]= \lim_{t\rightarrow \infty}
\frac{A_i (0,t)}{t}= \lambda_i,~\mbox{for all}~t\in \mathbb Z_+,$$
where $A_i(s,t)<\infty$ denotes the cumulative arrival to queue $i$ in the time interval $(s,t]$.
Each queue $i$ can be serviced at rate $c_i(t)\geq 0$ representing the
potential departure rate of work from the queue $Q_i(t).$ We consider
finite state Markov time-varying channels \cite{WM95FMC}: each
$\{\boc(t) = [c_i(t)]:t\geq 0\}$ is a continuous-time, time-homogeneous, irreducible Markov
process, where each link has $m$ states channel space such that
$c_i(t) \in \sch:=\{\chs_1,\dots,\chs_m \}$ and $0 < \chs_1 < \dots < \chs_m
=1.$ We denoted by $\ratech^{ \boldsymbol{u} \to \boldsymbol{v}}$ the
`transition-rates' on the channel state for $\boldsymbol{u} \to
\boldsymbol{v},$ $\boldsymbol{u} ,\boldsymbol{v} \in \sch^n.$ For the
time-varying channels, we assume that each link $i$ knows the channel
state $c_i(t)$ before it transmits.\footnote{The channel information can
  be achieved using control messages such as RTS and CTS in IEEE 802.11,
  and links can adapt their transmission parameters to channel
  transitions for every transmission by changing coding and modulation
  parameters.}  We call
$\max_{\boldsymbol{u} \in \sch^n} \{ \sum_{\boldsymbol{v} \in
  \sch^n:\boldsymbol{v} \neq \boldsymbol{u}} \ratech^{\boldsymbol{u} \to
  \boldsymbol{v}}\}$ the {\em channel varying speed.}  The inverse of
{\em channel varying speed} indicates the maximum of the expected number
of channel transitions during the unit-length time interval.
We consider only single-hop sessions (or flows), \ie, once work departs
from a queue, it leaves the network.

The queues are offered service as per the constraint imposed by
interference. To model this constraint, we adopt a popular graph-based
approach, where denote by $G = (V, E)$ the
inference graph among $n$ queues, 
where the vertices $V = \{1,\ldots,n\}$ represent
queues and the edges $E \subset V \times V$ represent
interferences between queues: $(i,j) \in E,$ if queues $i$ and $j$ interfere with
each other. Let $\mathcal N(i) = \{ j \in V : (i,j) \in E\}$ and
$\bsigma(t) = [\sigma_i(t)] \in \{0,1\}^n$ denote the neighbors of node
$i$ and a schedule at time $t$, \ie, whether queues transmit at time
$t,$ respectively, where $\sigma_i(t) = 1$ represents transmission of
queue $i$ at time $t$. Then, interference imposes the constraint that
for all $t \in \R_+$,
$\bsigma(t)  \in  \cI(G)$, where
 $$  \cI(G) := \big\{ \brho = [\rho_i] \in \{0,1\}^n:
    ~\rho_i + \rho_j \leq 1,  ~\forall (i,j) \in E \big\}. $$
The resulting queueing dynamics are described as follows. For 
$0 \leq s < t$ and $1\leq i\leq n$,
$$ Q_i(t) = Q_i(s) - \int_{s}^t \sigma_i(r) c_i(r) \bone_{\{Q_i(r) > 0\}} ~dr + A_i(s,t), $$
where 
$\bone_{\{\cdot\}}$ denotes the indicator function.
Finally, we define the cumulative actual and potential departure processes $\bD(t) = [D_i(t)]$ and
$\widehat \bD(t) = [\widehat D_i(t)]$, respectively, where
\begin{equation*}
D_i(t) = \int_{0}^t \sigma_i(r) c_i(r)\bone_{\{Q_i(r) > 0\}} dr, \ \widehat D_i(t) = \int_{0}^t \sigma_i(r) c_i(r) dr.
\end{equation*}

\subsection{Scheduling, Rate Region and Metric}
The main interest of this paper is to design a scheduling algorithm
which decides $\bsigma(t)\in \cI(G)$ for each time instance
$t\in\mathbb{R}_+$. Intuitively, it is expected that a good scheduling algorithm
will keep the queues as small as possible. To formally discuss, we
define the maximum achievable rate region (also called capacity region) $\cH \subset [0,1]^n$ of the network,
which
is the convex hull of the feasible scheduling set $\cI(G)$, \ie,
\begin{multline*}
\cH = \cH(\bgamma, G)=\Big\{ \sum_{\bc\in\sch^n} \pi_{\bc} \sum_{\brho \in \cI(G)}\alpha_{\brho,\bc} \bc^{T}\cdot \brho ~:
  \alpha_{\brho,\bc} \ge 0 
\text{ and } \cr \sum_{\brho \in
    \cI(G)}{\alpha_{\brho,\bc}} = 1\text{ for all }\bc\in\sch^n
\Big\}, ~~ 
\end{multline*}
where $\bc^{T}\cdot \brho=[c_i\rho_i]$ and $\pi_{\bc}$ denotes the stationary distribution of
channel state $\bc$ under the channel-varying Markov process.
The intuition behind this definition comes from the facts: (a) any
scheduling algorithm has to choose a schedule from $\cI(G)$ at each
time and channel state where $\alpha_{\brho,\bc}$ denotes the fraction
of time selecting schedule $\brho$ for given channel state $\bc$ and (b)
for channel state $\bc\in\sch^n$, the fraction in the time domain
where $\bc(t)=[c_i(t)]$ is equal to $\bc$ is
$\pi_{\bc}.$ 
Hence the time average of the `service rate' induced by any
algorithm must belong to $\cH$. 

We call the arrival rate $\blambda$ {\em admissible} if $\blambda =[\lambda_i]\in\bLambda=\bLambda(\bgamma, G)$, where 
$$ \bLambda(\bgamma, G):= \left\{ \blambda \in \mathbb{R}^n_{+} :
  \blambda \leq \blambda^{\prime}\text{, for some } \blambda^{\prime}\in
  \cH(\bgamma, G) \right\},$$
where $\blambda \leq \blambda^{\prime}$
corresponds to the component-wise inequality, \ie, if
$\blambda\notin\bLambda$, queues should grow linearly over time under
any scheduling algorithm.  Further, 
$\blambda$ is called {\em strictly
  admissible} if $\blambda \in\bLambda^o=\bLambda^o(\bgamma, G)$ and
$$ \bLambda^o(\bgamma, G) := \left\{ \blambda \in \mathbb{R}^n_{+} :
\blambda < \blambda^{\prime} \text{, for some } \blambda^{\prime}\in \cH(\bgamma, G) \right\}. $$
We now define the performance metric. 
\begin{definition}\label{def:sched-rate-regi}
A scheduling algorithm is called rate-stable for a given arrival rate $\blambda$, if
  \begin{equation}\label{eq1}\lim_{t \rightarrow \infty} \frac{1}{t}\bD(t)=\blambda\qquad(\mbox{with probability }1).\end{equation}
Furthermore, we say a scheduling algorithm has $\alpha$-throughput if it is rate-stable for 
	any $\blambda\in\alpha\bLambda^o(\bgamma,G)$. In particular, when $\alpha=1$, it is called throughput optimal.	
\end{definition}

We note that \eqref{eq1} is equivalent to 
$ \lim_{t\rightarrow \infty} \frac{1}{t}\Q(t) = 0,$
since $\lim_{t\to\infty} \frac{A_i(0,t)}{t}  =  \lambda_i$ (because
the arrival process is stationary ergodic).
The following lemma implies that the potential departure process suffies to study
the rate-stability.
\begin{lemma}\label{lem2}
  A scheduling algorithm is rate-stable if 
  $$\lim_{t \rightarrow \infty} \frac{1}{t}\widehat{\bD}(t)>\blambda.$$
\end{lemma}
We omit the proof due to space limitation. 

\subsection{Channel-aware CSMA Algorithms: A-CSMA}

The algorithm to decide $\bsigma(t)$ utilizing the local carrier
sensing information can be classified as CSMA (Carrier Sense Multiple
Access) algorithms. 
In between two transmissions, a queue waits for a random amount of time
-- also known as {\em backoff time}. Each queue can sense the medium
perfectly and instantly, \ie, knows if any other interfering queue is
transmitting at a given time instance. If a queue that finishes waiting
senses the medium to be busy, it starts waiting for another random
amount of time; else, it starts transmitting for a random amount of
time, called {\em channel holding time}. 
We assume that queue $i$'s backoff and channel holding times have
exponential distributions with mean $1/R_i$ and $1/S_i$, respectively,
where $R_i=R_i(t)>0$ and $S_i=S_i(t)>0$ may change over time. We define
A-CSMA (channel-aware CSMA) to be the class of CSMA algorithms where
$R_i(t)$ and $S_i(t)$ are decided by some functions of the current
channel capacity, \ie, $R_i(t)= f_i(c_i(t))$ and $S_i(t)=g_i(c_i(t))$
for some functions $f_i$ and $g_i$. In the special case when $R_i(t)$
and $S_i(t)$ are decided independently of current channel information
(\eg, $f_i$'s and $g_i$'s are constant functions),
we specially say a CSMA algorithm is U-CSMA (channel-unaware CSMA).



Then, given functions $[f_i]$ and $[g_i]$, it is easy to check that\\ 
$\{(\bsigma(t),\bc(t)): t\geq 0\}$ under A-CSMA is a continuous time
Markov process, whose kernel (or transition-rates) is given by: 
\begin{eqnarray}
\label{eq:ucsma_kernel}
(\bsigma,\boldsymbol{u}) &\rightarrow& (\bsigma,\boldsymbol{v})~~\text{with
  rate}~~\ratech^{\boldsymbol{u} \to \boldsymbol{v}} \cr
(\bsigma^0_i,\bc) &\rightarrow& (\bsigma^1_i,\bc)~~\text{with
  rate}~~f_i(c_i) \cdot \prod_{j:(i,j)\in E}(1-\sigma_j)\cr
(\bsigma^1_i,\bc) &\rightarrow& (\bsigma^0_i,\bc)~~\text{with
  rate}~~g_i(c_i) \cdot \sigma_i,
\end{eqnarray}
where $\bsigma^0_i$ and $\bsigma^1_i$ denote two `almost' identical schedule
vectors except $i$-th elements which are 0 and 1,
respectively. Since $\{\bc(t)\}$ is a time-homogeneous irreducible Markov process,
$\{(\bsigma(t),\bc(t))\}$ is ergodic, \ie, it has the unique stationary distribution
$[\pi_{\bsigma,\bc}]$. For example, when functions $f_i$ and $g_i$ are constant (\ie, 
U-CSMA with fixed $R_i(t)=R_i$ and $S_i(t)=S_i$), 
$$\pi_{\bsigma,\bc}~=~\pi_{\bc} \cdot
\frac{\exp\left(\sum_{i} {\sigma_i} \log \frac{R_i}{S_i}\right)}{\sum_{\brho=[\rho_i]\in\cI(G)}\exp\left(\sum_{i} {\rho_i} \log \frac{R_i}{S_i}\right)},$$
and if $\{\bc(t)\}$ is (time-)reversible,
$\{(\bsigma(t),\bc(t))\}$ is as well.
In general, $\{(\bsigma(t),\bc(t))\}$ is not reversible unless functions $f_i/g_i$ are constant, 
as we state in the following theorem.

\begin{theorem}\label{thm:nonrev}
If $\{(\bsigma(t),\bc(t))\}$ is reversible,
$$\frac{f_i(x)}{g_i(x)}=\frac{f_i(y)}{g_i(y)},
\qquad\mbox{for all}~x,y \in
\sch,i\in V.$$
\end{theorem}
\begin{proof}
  We prove this by contradiction.  Denote by $\bc^u_i$ and $\bc^v_i$ two
  almost identical channel state vectors except $i$-th elements, which
  are $\chs_u$ and $\chs_v,$ respectively.  Suppose that $\{(\bsigma(t),\bc(t))\}$ 
is reversible and $\frac{f_i(\chs_u)}{g_i(\chs_u)}
  \neq \frac{f_i(\chs_v)}{g_i( \chs_v)}$ for some link $i$. From the reversibility, the
  transition path $(\bsigma_i^0,\bc^u_i) \to (\bsigma_i^0,\bc^v_i) \to
  (\bsigma_i^1,\bc^v_i)$ has to satisfy the following balance equations:
\begin{align}\label{eq:b1}
&\pi_{\bsigma_i^0,\bc^u_i}\ratech^{\bc^u_i \to \bc^v_i}  =
\pi_{\bsigma_i^0,\bc^v_i}\ratech^{\bc^v_i \to \bc^u_i}
\cr &\pi_{\bsigma^0_i,\bc_i^v} f_i(\chs_v)  =
\pi_{\bsigma^1_i,\bc_i^v} g_i(\chs_v),\end{align}
Similarly, for the transition path $(\bsigma_i^0,\bc^u_i) \to
(\bsigma_i^1,\bc^u_i) \to (\bsigma_i^1,\bc^v_i)$,
\begin{align} \label{eq:b2}
&\pi_{\bsigma_i^0,\bc^u_i}f_i(\chs_u)  =
\pi_{\bsigma_i^1,\bc^u_i}g_i(\chs_u), ~\mbox{and}~\cr &\pi_{\bsigma^1_i,\bc_i^u}\ratech^{\bc_i^u \to \bc_i^v}  =  \pi_{\bsigma^1_i,\bc_i^v}\ratech^{\bc_i^v \to \bc_i^u}.
\end{align}
From \eqref{eq:b1} and \eqref{eq:b2}, 
\begin{eqnarray}\label{eq:b3}
\frac{\pi_{\bsigma_i^0,\bc^u_i}}{\pi_{\bsigma_i^1,\bc^v_i}} =
\frac{\ratech^{\bc_i^v \to \bc_i^u}  g_i(\chs_v)}{\ratech^{\bc_i^u \to
    \bc_i^v} f_i(\chs_v)} =\frac{\ratech^{\bc_i^v \to \bc_i^u}g_i(\chs_u)}{\ratech^{\bc_i^u \to
    \bc_i^v}f_i(\chs_u)},
\end{eqnarray}
which contradicts the assumption $\frac{f_i(\chs_u)}{g_i(\chs_u)} \neq
\frac{f_i(\chs_v)}{g_i(\chs_v)}$. This completes the proof of Theorem \ref{thm:nonrev}.
\end{proof}

We note that the non-reversible property makes it hard to characterize
the stationary distribution $[\pi_{\bsigma,\bc}]$ of the Markov process induced by A-CSMA.

\section{Achievable Rate Region of A-CSMA}

In this section, we study the achievable rate region of
A-CSMA algorithms given (fixed) functions $[f_i]$ and $[g_i]$.
We show that the achievable rate region of A-CSMA is maximized for the
following choices of functions: 
\begin{equation}\label{eq:csma}
  \log\frac{f_i (x)}{g_i(x)}~=~r_i \cdot x,~~\mbox{for}~x\in [0,1],\end{equation}
where $r_i\in \mathbb R$ is some constant. 
Namely, the ratio ${f_i (x)}/{g_i(x)}$ is an exponential function in terms of $x$.
We let EXP-A-CSMA denote the sub-class of A-CSMA algorithms with functions satisfying \eqref{eq:csma} for some $[r_i]$.
The following theorem justifies the optimality of EXP-A-CSMA in terms of its
achievable rate region.
\begin{theorem}[Optimality]\label{thm3}
  For any arrival rate $\blambda
  =[\lambda_i]\in \bLambda^o,$ interference graph $G$,
  and channel transition-rate $\bgamma$, there exists $[r_i]$, $[f_i]$ and $[g_i]$ satisfying \eqref{eq:csma}
  such that the corresponding EXP-A-CSMA algorithm is rate-stable. 
\end{theorem}

We also establish that Theorem \ref{thm3} is tight in the sense that it
does not hold for other A-CSMA algorithms that have different ways of
reflecting channel capacity in adjusting CSMA parameters. 
To state it formally, given a
non-negative continuous function $k: [0,1]\rightarrow \mathbb R_+$, we
define EXP($k$)-A-CSMA as the sub-class of A-CSMA algorithms with
the following form of functions:
\begin{equation}  \log\frac{f_i (x)}{g_i(x)}~=~r_i \cdot k(x),~~\mbox{for}~x\in [0,1],\label{eq:kcsma}\end{equation}
where $r_i\in \mathbb R$ is some constant.
The following theorem states that EXP-A-CSMA is the unique class of A-CSMA
maximizing its achievable rate region.
\begin{theorem}[Uniqueness]\label{thm:nonlinear}
If the conclusion of Theorem \ref{thm3} holds for EXP($k$)-A-CSMA, then 
$$\mbox{EXP($k$)-A-CSMA}~=~\mbox{EXP-A-CSMA.}$$
\end{theorem}
The proofs of Theorems~\ref{thm3} and \ref{thm:nonlinear} are given
in Sections \ref{sec:pfthm3} and \ref{sec:pfthm:nonlinear}, respectively.
For the proof of Theorem \ref{thm3},
Sections \ref{sec:pflemthropt} and 
\ref{sec:pflemmas} describe the proofs of necessary lemmas, 
Lemmas~\ref{lem:thropt} and \ref{lem:thrwithin}.
In the following proofs (and throughout this paper), we commonly let 
$[\pi_{\bsigma,\bc}]$, $[\pi_{\bc}]$ and
$[\pi_{\bsigma|\bc}]$ be the stationary distributions of Markov
processes $\{(\bsigma(t),\bc(t))\}$, $\{\bc(t)\}$ and
$\{\bsigma(t),\bc\}$ induced by an A-CSMA algorithm, respectively.

\subsection{Proof of Theorem \ref{thm3}}\label{sec:pfthm3}

To begin with, we recall that
the channel varying speed $\psi$ is defined as:
$\psi = \max_{\boldsymbol{u} \in \sch^n} \{ \sum_{\boldsymbol{v} \in
  \sch^n:\boldsymbol{v} \neq \boldsymbol{u}} \ratech^{\boldsymbol{u} \to
  \boldsymbol{v}}\}.$
We first state Lemmas~\ref{lem:thropt} and \ref{lem:thrwithin}, which
are the key lemmas to the proof of Theorem \ref{thm3}.
\begin{lemma}\label{lem:thropt}
	For any $\delta_1\in(0,1)$, arrival rate $\blambda=[\lambda_i]\in(1-\delta_1)\bLambda^o$,
	 interference graph $G$ and channel transition-rate $\bgamma$, 
	there exists $[r_i]\in \mathbb R^n$ such that 
	$$\max_i \left|r_i \right|~\leq~
	\frac{4n^2\log |\cI(G)|}{\delta_1^2 \min\limits_i
          \{\left(\sum_{\bc \in \sch^n} c_i \pi_{\bc}\right)^2\}},
	$$
	and every EXP-A-CSMA algorithm with 
	$$\log \frac{f_i(\chs)}{g_i( \chs)}= r_i \cdot \chs, \quad
        \mbox{ for all $i\in V, \chs\in \sch$}$$
	satisfies 
	\begin{align*}
&\lambda_{i} ~\leq~ \sum_{\bc \in \sch^n}c_i\pi_{\bc}\sum_{\bsigma \in \cI(G): \sigma_i = 1}\pi_{\bsigma|\bc},
\qquad\mbox{for all}~i\in V.\qquad\qquad\qquad
\end{align*}
\end{lemma}
\begin{lemma}\label{lem:thrwithin}
	For any $\delta_2\in(0,1)$, interference graph $G$ and channel transition-rate $\bgamma$ and
	 A-CSMA algorithm with functions $\bof =[f_i]$ and $\bog=[g_i]$
         satisfying 
	$$\min_{i\in V,\chs \in \sch}\{f_i(\chs) , g_i(\chs)\}\geq
        \frac{\psi \cdot m^{2^{n}m^n(n+1)} }{\delta_2},$$ it follows that
	$$\max_{(\bsigma,\bc)\in \mathcal I(G)\times\sch^n}\left|1-\frac{\pi_{\bsigma,\bc}}{\pi_{\bc}\pi_{\bsigma|\bc}}\right|~ <~ \delta_2.$$
\end{lemma}

Lemma \ref{lem:thrwithin} implies that if $f_i,g_i$ are large enough,
the stationary distribution $[\pi_{\bsigma,\bc}]$
approximates to a product-form distribution $[\pi_{\bc}\pi_{\bsigma|\bc}]$, 
where under EXP-A-CSMA,
$$\pi_{\bsigma|\bc}~\propto~\exp\left(\sum_i \sigma_i r_i c_i\right),$$
due to the reversibility of Markov process $\{\bsigma(t),\bc\}$.
On the other hand, Lemma \ref{lem:thropt} implies that
arrival rate $\blambda$ is stabilized under the distribution $[\pi_{\bc}\pi_{\bsigma|\bc}]$. 
Therefore, combining two above lemmas will lead to the proof of
Theorem \ref{thm3}.


We remark that Lemma \ref{lem:thropt} is a non-trivial generalization of Lemma 8 in
\cite{JSSW10DRA} (for static channels), which corresponds to a special
case of Lemma \ref{lem:thropt} with $\pi_{\bc}=1$ for $\bc=[1]$. 


\smallskip
\noindent {\bf \em Proof of Theorem \ref{thm3}.}
We now complete the proof of Theorem \ref{thm3} using
Lemmas \ref{lem:thropt} and \ref{lem:thrwithin}.  For a given
arrival rate $\blambda\in \bLambda^o$, 
there exists $\varepsilon\in(0,1)$ such that $\blambda\in
(1-\varepsilon)\bLambda^o$ since $\blambda\in \bLambda^o$.  If we apply
Lemmas \ref{lem:thropt} and \ref{lem:thrwithin} with
$(1+\varepsilon)\blambda\in(1-\varepsilon^2)\bLambda^o$ (\ie, 
$\delta_1=\varepsilon^2$ and
$\delta_2=\frac{\varepsilon}{1+\varepsilon}$), we have that there exists
an EXP-A-CSMA algorithm with constant $[r_i]$ and functions $[f_i]$ and
$[g_i]$ such that
\begin{eqnarray*}
	\eta&\leq& \min_{i\in V,\chs\in \sch} \{f_i(\chs), g_i(\chs)\} \\
(1+\varepsilon)\lambda_{i} &\leq& \sum_{\bc \in \sch^n}c_i\pi_{\bc}\sum_{\bsigma \in \cI(G): \sigma_i = 1}\pi_{\bsigma|\bc},
\end{eqnarray*}
where we choose $$f_i(c_i)=R= \eta\exp(\kappa),~~
g_i(c_i)=R\cdot\exp(-r_i\cdot c_i),$$
$$\kappa=\kappa(\delta_1,G,\bgamma):= 
	\frac{4n^2\log |\cI(G)|}{\delta_1^2 \min\limits_i
          \{\left(\sum_{\bc \in \sch^n} c_i \pi_{\bc}\right)^2\}}, $$
and
$$\eta=\eta(\delta_2,G,\bgamma):=\frac{\psi \cdot m^{2^{n}m^n(n+1)} }{\delta_2}. $$
Therefore, it follows that
\begin{eqnarray*}
	\lambda_{i} &\leq&\left(1-\frac{\varepsilon}{1+\varepsilon}\right)\sum_{\bc \in \sch^n}c_i\pi_{\bc}\sum_{\bsigma \in \cI(G): \sigma_i = 1}\pi_{\bsigma|\bc}\\
	&<&~ \sum_{\bc\in\sch^n}c_i\pi_{\bc}\sum_{\bsigma \in \cI(G): \sigma_i = 1}\pi_{\bsigma,\bc}\\
	&=&~\lim_{t\to\infty} \frac1t \widehat D_{i}(t),
\end{eqnarray*}
where the last inequality is from the ergodicity of Markov process $\{(\bsigma(t),\bc(t))\}$.
This leads to the rate-stability using Lemma \ref{lem2}, and hence
completes the proof.

\subsection{Proof of Lemma \ref{lem:thropt}}\label{sec:pflemthropt}  
We use a similar strategy with that of Lemma 8 in \cite{JSSW10DRA}. 
Since $\blambda\in (1-\delta_1)\bLambda^o$, there exists $\blambda^{\prime}=[\lambda^{\prime}_i]\in (1-\delta_1/2)\cH$ 
such that $\blambda\leq \blambda^{\prime}$ and 
$$\lambda^{\prime}_i\geq \frac{\delta_1}{2n}\cdot \sum_{\bc \in \sch^n}c_i\pi_{\bc} \qquad\mbox{for all}~i\in V.$$
For such a choice of $\blambda^{\prime}$, we consider the following function $F:\mathbb R^n\to\mathbb R$:
\begin{equation*}
  F(\bm{r}) = \bm{\lambda}^{\prime}\cdot\bm{r} - \sum_{\bc \in \sch^n }\pi_{\bc}\log \left( \sum_{\bsigma \in \cI(G)} \exp\left(\sum_{i}\sigma_i c_i r_i\right) \right).
\end{equation*}

One can easily check that
$F$ is strictly concave and bounded above.
Hence, there exists a unique maximizer $\bm{r}^*\in\mathbb R^n$ 
such that $F(\bm{r}^*)=\sup_{\bm{r}\in\mathbb R^n}F(\bm{r})$ and $\nabla F(\bm{r}^*)=0$. We prove the following.
\begin{eqnarray}
	\max_i r_i^*&\leq& \frac{2\log |\cI(G)|}{\delta_1 \min\limits_i \left\{\sum_{\bc \in \sch^n} c_i \pi_{\bc}\right\}}\label{eq1:pflemthropt}\\
	\min_i r_i^*&\geq & -\frac{4n^2\log
          |\cI(G)|}{\delta_1^2\min\limits_i \left\{\left(\sum_{\bc \in \sch^n} c_i \pi_{\bc}\right)^2\right\}}\label{eq2:pflemthropt}
\end{eqnarray}

\noindent {\bf Proof of (\ref{eq1:pflemthropt}).} 
Suppose there exists $i$ such that 
$$r_i> \frac{2\log |\cI(G)|}{\delta_1 \sum_{\bc \in \sch^n} c_i \pi_{\bc}}.$$ Then, $\bm{r}$ cannot be a maximizer of $F$ since
\begin{eqnarray*}  
	F(\bm{r})& =& \bm{\lambda}^{\prime}\cdot\bm{r} - \sum_{\bc \in
          \sch^n }\pi_{\bc}\log \left( \sum_{\bsigma \in
            \cI(G)} \exp\left(\sum_{i}\sigma_i c_i r_i\right) \right) \\
 &= & \sum_{\bc \in \sch^n }\pi_{\bc}\sum_{\brho \in \cI(G)}\widehat
 \pi_{\brho|\bc}\log \frac{\exp\left(\sum_{i}\rho_i c_i r_i\right)}{
     \sum_{\bsigma \in \cI(G)} \exp(\sum_{i}\sigma_i c_i r_i)}\\
 &\le& \sum_{\bc \in \sch^n }\pi_{\bc}\widehat \pi_{\bm{0}|\bc} \log
 \frac{\exp(0)}{ \sum_{\bsigma \in \cI(G)} \exp(\sum_{i}\sigma_i
     c_i r_i) }\\
 &\le& -\sum_{\bc \in \sch^n }\pi_{\bc}\widehat \pi_{\bm{0}|\bc} c_i r_i ~ \le~ -\frac{\delta_1}2 \cdot r_i \cdot \sum_{\bc \in \sch^n} c_i \pi_{\bc} \\
 &<& -\log|\cI(G)|~\le~ F(\bm{0}) ~\le~ \sup_{\br\in\mathbb R^n} F(\br).
\end{eqnarray*}
In above, for the second equality,  
there exist a non-negative valued measure $\widehat \pi$ such that for all $\bc \in \sch^n$
\begin{align}
\blambda^{\prime} = \sum_{\bc \in \sch^n}\pi_{\bc}\sum_{\bsigma \in \cI(G)}\widehat \pi_{\bsigma|\bc}[c_i\sigma_i]_{1 \leq i \leq n} ~~ \text{and} ~~ \widehat \pi_{\bm{0}|\bc} \ge \frac{\delta_1}2 
\end{align}
since $\blambda^{\prime} \in (1-\delta_1/2)\cH$.
This completes the proof of (\ref{eq1:pflemthropt}).

\noindent {\bf Proof of (\ref{eq2:pflemthropt}).}
From (\ref{eq1:pflemthropt}) it suffices to prove that 
$\bm{r}$ cannot be a maximizer of $F$ if there exists $j$ such that,
for all $k\neq j,$
$$r_j<-\frac{4n^2\log |\cI(G)|}{\delta_1^2\min\limits_i \left(\sum_{\bc \in \sch^n} c_i \pi_{\bc}\right)^2}~~\mbox{and}~~
r_k\leq \frac{2\log |\cI(G)|}{\delta_1\min\limits_i \{ \sum_{\bc \in
    \sch^n} c_i \pi_{\bc}\}}.$$
The proof is completed by the following:
\begin{eqnarray*}  
	F(\bm{r})& =& \bm{\lambda}^{\prime}\cdot\bm{r} - \sum_{\bc \in
          \sch^n }\pi_{\bc}\log \left( \sum_{\bsigma \in \cI(G)}
          \exp\left(\sigma_i c_i r_i\right) \right) \\
 &\leq &\bm{\lambda}^{\prime}\cdot\bm{r}\\
&\leq&(n-1)\cdot\frac{2\log |\cI(G)|}{\delta_1\min\limits_i \{ \sum_{\bc \in \sch^n} c_i \pi_{\bc} \}}+\lambda_j^{\prime}r_j\\
&\leq&(n-1)\cdot\frac{2\log |\cI(G)|}{\delta_1\min\limits_i \{ \sum_{\bc \in \sch^n} c_i \pi_{\bc} \}}+
\frac{\delta_1}{2n}\cdot  \sum_{\bc \in \sch^n} c_j \pi_{\bc} \cdot r_j\\
&<& -\log|\cI(G)|\le F(\bm{0})
\le \sup_{\br\in\mathbb R^n} F(\br).
\end{eqnarray*}

Then, from (\ref{eq1:pflemthropt}) and (\ref{eq2:pflemthropt}),
$$\max_i\left|r^*_i\right|\leq 
\frac{4n^2\log |\cI(G)|}{\delta_1^2\min\limits_i \left( \sum_{\bc \in \sch^n} c_i \pi_{\bc} \right)^2}.$$
Furthermore, computing the first derivative of $F$ gives us
\begin{align*}
0=\frac{\partial}{\partial r_i}F(\bm{r}^*)&=\lambda_i^{\prime}-
\sum_{\bc \in
  \sch^n}c_i\pi_{\bc}\frac{\sum\limits_{\bsigma\in\cI(G)}\sigma_i 
  \exp\left(\sum_{i}\sigma_i c_i r_i^*\right)}{\sum\limits_{\brho \in \cI(G)}
  \exp\left(\sum_{i}\rho_i c_i r_i^*\right)}\\
&=\lambda_i^{\prime}-
\sum_{\bc \in
  \sch^n}c_i\pi_{\bc}\sum_{\bsigma \in \cI(G): \sigma_i =
    1}\pi_{\bsigma|\bc},
\end{align*}
where we now choose an EXP-A-CSMA algorithm such that
$\log \frac{f_i(\chs)}{g_i(\chs)}=r^*_i \cdot \chs$. 
Therefore, it follows that 
\begin{eqnarray*}
  \lambda_i\leq\lambda_i^{\prime}=\sum_{\bc \in \sch^n}c_i \pi_{\bc}\sum_{\bsigma \in \cI(G): \sigma_i =
    1}\pi_{\bsigma|\bc} .
\end{eqnarray*}
This completes the proof of Lemma \ref{lem:thropt}.

\subsection{Proof of Lemma \ref{lem:thrwithin}}\label{sec:pflemmas}

Let $\mathcal{G}=(\mathcal{V},\mathcal{E})$ denote a weighted directed graph induced by
Markov process $\{(\bsigma(t),\bc(t))\}$: 
$\mathcal{V}=\cI(G)\times \sch^n$ and 

\noindent$((\bsigma_1,\bc_1), (\bsigma_2,\bc_2))\in\mathcal{E}$ if 
the transition-rate (which becomes the weight of the edge) from $(\bsigma_1,\bc_1)$ to 
$(\bsigma_2,\bc_2)$ is non-zero in Markov process $\{(\bsigma(t),\bc(t))\}$.
Hence, there are two types of edges:
\begin{itemize}
	\item[I.] $((\bsigma_1,\bc_1),(\bsigma_2,\bc_2))\in \mathcal{E}$ and $\bsigma_1=\bsigma_2$
	\item[II.] $((\bsigma_1,\bc_1),(\bsigma_2,\bc_2))\in \mathcal{E}$ and $\bc_1=\bc_2$
\end{itemize}
A subgraph of $\mathcal G$ is called \textit{arborescence} (or spanning tree) with root $(\bsigma,\bc)$
if for any vertex in $\mathcal V\setminus\{(\bsigma,\bc)\}$, there is exactly one directed path from 
the vertex to root $(\bsigma,\bc)$ in the subgraph.
Let $\mathcal{A}_{\bsigma,\bc}$ and
$w(\mathcal{A}_{\bsigma,\bc})$ denote the set of
\textit{arborescence}s of which root is $(\bsigma,\bc)$ and the 
sum of weights of \textit{arborescence}s in $\mathcal{A}_{\bsigma,\bc}$, where the weight of an
\textit{arborescence} is the product of weight of edges. Then, 
Markov chain tree theorem \cite{AT89PMCT} implies that
\begin{equation}
\pi_{\bsigma,\bc} =\frac{ w(\mathcal{A}_{\bsigma,\bc})}{\sum\limits_{(\brho,\bm{d}) \in \cI(G)\times \sch^n}
  w(\mathcal{A}_{\brho,\bm{d}})}.\label{eq:pipi}
\end{equation}

Now we further classify the set of \textit{arborescence}s. 
We let
$\mathcal{A}_{\bsigma,\bc}^{(i)}\subset\mathcal{A}_{\bsigma,\bc}$ denote the set of
\textit{arborescence}s consisting of $i$ edges of type I. Then, we have
\begin{align*}
& w(\mathcal{A}_{\bsigma,\bc}) =  \sum_{i \geq m^n-1}w(\mathcal{A}_{\bsigma,\bc}^{(i)}) \stackrel{(a)}{\le}w(\mathcal{A}_{\bsigma,\bc}^{(m^n -1)})+\\
& \qquad \sum_{i \ge m^n}\left(\frac{\delta_2}{m^{2^{n}m^n(n+1)} }\right)^{i+1-m^n}\cdot |\mathcal{A}_{\bsigma,\bc}^{(i)}| \cdot w(\mathcal{A}_{\bsigma,\bc}^{(m^n-1)})\\
 &\stackrel{}{\le}  w(\mathcal{A}_{\bsigma,\bc}^{(m^n-1)})\cdot \left( 1 + \sum_{i \ge m^n}
\left(\frac{\delta_2}{m^{2^{n}m^n(n+1)} }\right)^{i+1-m^n}\cdot |\mathcal{A}_{\bsigma,\bc}^{(i)}| \right)\\
 &{\le}  w(\mathcal{A}_{\bsigma,\bc}^{(m^n-1)})\cdot \left( 1 + \frac{\delta_2}{m^{2^{n}m^n(n+1)} }\cdot |\mathcal{A}_{\bsigma,\bc}| \right)\\
& \stackrel{(b)}{<} w(\mathcal{A}_{\bsigma,\bc}^{(m^n-1)})\cdot ( 1 + \delta_2 ),
\end{align*}
where $(a)$ is from the condition in Lemma \ref{lem:thrwithin} and
for $(b)$ we use the inequality $|\mathcal{A}_{\bsigma,\bc}| < (mn)^{2^{n}m^n} $. 
Therefore, using the above inequality, it follows that
\begin{eqnarray*}
	\frac{\pi_{\bsigma,\bc}}{\pi_{\bc}\pi_{\bsigma|\bc}} &=& \frac{ w(\mathcal{A}_{\bsigma,\bc})}{ w(\mathcal{A}_{\bsigma,\bc}^{(m^n-1)})}
	\cdot \frac{\sum\limits_{\bm{d} \in \sch^n}\sum\limits_{\brho \in \cI(G)}w(\mathcal{A}_{\brho,\bm{d}}^{(m^n-1)})}{\sum\limits_{\bm{d} \in \sch^n}\sum\limits_{\brho \in \cI(G)}w(\mathcal{A}_{\brho,\bm{d}})}\\
	&<&1+\delta_2, 
\end{eqnarray*}
where the first equality follows from (\ref{eq:pipi}) and 
$$\pi_{\bc}\pi_{\bsigma|\bc} =\frac{ w(\mathcal{A}_{\bsigma,\bc}^{(m^n-1)})}{\sum\limits_{\bm{d} \in \sch^n}
\sum\limits_{\brho \in \cI(G)}w(\mathcal{A}_{\brho,\bm{d}}^{(m^n-1)})}.$$
Similarly, one can also show that 
$\frac{\pi_{\bsigma,\bc}}{\pi_{\bc}\pi_{\bsigma|\bc}}~>~ 1-\delta_2.$
This completes the proof of Lemma \ref{lem:thrwithin}. 

\subsection{Proof of Theorem~\ref{thm:nonlinear}}\label{sec:pfthm:nonlinear}

Consider a star interference graph $G$, where $1$ denotes the center vertex and
$\{2,\dots,n\}$ is the set of other outer vertices. For the time-varying
channel model, we set each element of channel by $\chs_j =
\frac{j}{m}$ and assume the channel transition satisfies $\pi_{\bc} =
\frac{1}{m^n}.$ For the arrival rate, we choose $\blambda =
\arg\max\limits_{\blambda \in (1-\varepsilon)\cH} \sum_{i} \lambda_i$,
where $\varepsilon\in(0,1)$ will be chosen later.

Under this setup, 
suppose the conclusion of Theorem \ref{thm3} holds, \ie, 
there exists a rate-stable EXP($k$)-A-CSMA. 
Then, from the ergodicity of Markov process $\{(\bsigma(t),\bc(t))\}$, we have
\begin{equation}\label{eq:max}
	\lambda_{i} ~\leq~ \sum_{\bc \in \sch^n}\pi_{\bc}\sum_{\bsigma \in \cI(G)}c_i\sigma_i\pi_{\bsigma|\bc},
\qquad\mbox{for all}~i\in V. \end{equation}
Taking the summation over $i\in V$ in both sides of the above inequality and using
$\blambda  = \arg\max\limits_{\blambda \in (1-\varepsilon)\cH} \sum_{i}
\lambda_i$, it follows that
\begin{align*}
	\sum_{i}
	\lambda_i&~=~(1-\varepsilon)\sum_{\bc \in
    \sch^n}\pi_{\bc} \max_{\brho\in \cI(G)}\sum_{i}c_{i}\rho_i \\
&~\leq~ \sum_{\bc \in \sch^n}\pi_{\bc}\sum_{\bsigma \in \cI(G)}\pi_{\bsigma|\bc}\sum_ic_i\sigma_i. \end{align*}
By rearranging terms in the above inequality, we have
\begin{align*}
\frac{\varepsilon}{1-\varepsilon} \sum_i \lambda_i&\geq\sum_{\bc \in
\sch^n}\pi_{\bc} \left(\max_{\brho\in \cI(G)}\sum_{i}c_{i}\rho_i
-\sum_{\bsigma \in \cI(G)}\pi_{\bsigma|\bc}\sum_ic_i\sigma_i\right)\\
&=\sum_{\bc \in \sch^n}\pi_{\bc}\cdot
E\left[\max_{\brho \in \cI(G)} \sum_i c_i\rho_i   - \sum_ic_i\sigma_i\right],
\end{align*}
where the expectation is taken with respect to random variable $\bsigma=[\sigma_i]$ of which distribution is $[\pi_{\bsigma|\bc}]$.
Since we know $\max_{\brho \in \cI(G)} \sum_i c_i\rho_i   - \sum_ic_i\sigma_i\geq 0$ with probability 1, we further have
that for all channel state $\bc,$
\begin{eqnarray*}
E\left[\max_{\brho \in \cI(G)} \sum_i c_i\rho_i   - \sum_ic_i\sigma_i\right]&\le& \frac{\sum_{i} \lambda^{\max}_i \cdot \frac{\varepsilon}{1-\varepsilon}}{\pi_{\bc}} \\
&\le&\frac{n \cdot \frac{\varepsilon}{1-\varepsilon}}{\pi_{\bc}}= n \cdot m^n \cdot \frac{\varepsilon}{1-\varepsilon}.
\end{eqnarray*}
Markov's inequality implies that 
$$\Pr\left[\max_{\brho \in \cI(G)} \sum_ic_i\rho_i  - \sum_{i}
c_i\sigma_i  \ge \frac{1}{m}\right] ~\le~ n \cdot m^{n+1} \cdot \frac{\varepsilon}{1-\varepsilon}.$$
If we choose $\varepsilon = \frac{1}{4n \cdot m^{n+1}},$ then 
\begin{equation}\Pr\left[\max_{\brho
  \in \cI(G)} \sum_ic_i\rho_i - \sum_ic_i\sigma_i \ge
\frac{1}{m}\right] ~<~ \frac{1}{2}.\label{eq:epsil}\end{equation} 

In the star graph $G$, $\max_{\bsigma \in \cI(G)} \sigma_i c_i$ is $c_1$
or $\sum_{j=2}^{n}c_j$ and under the channel model, if $c_1 \neq
\sum_{j=2}^{n}c_j,$ $\left|c_1 - \sum_{j=2}^{n}c_j \right| \ge
\frac{1}{m}.$ If $c_1 > \sum_{i=2}^n c_i,$ $r_i$
Thus, from \eqref{eq:epsil},
if $c_1 > \sum_{i=2}^n c_i,$$P[\sigma_1 = 1] > \frac{1}{2}$ and
if $c_1 < \sum_{i=2}^n c_i,$ $P[\sigma_1 = 1] < \frac{1}{2}$, which
implies that
$r_1 k(c_1) \ge \sum_{i=2}^n  r_i k(c_i)$ and $r_1 k(c_1) \le
\sum_{i=2}^n r_i k(c_i),$ respectively. Therefore, for every channel
state $\bc$ with $0<\sum_{i=2}^{n}c_i<1,$
\sqeq
\begin{equation}r_1 k\left( \sum\limits_{i=2}^{n}c_i+ \frac{1}{m}
  \right) ~>~ \sum\limits_{i=2}^n r_i k(c_i) ~>~
  r_1k\left(\sum\limits_{i=2}^{n}c_i - \frac{1}{m}
  \right), \label{eq:lin} \end{equation}\unsqeq
which implies that $r_ik(x)$
is a strictly increasing function. In addition, $r_1k(c_1) >0$ from
\eqref{eq:epsil}, because $\sum_{\bsigma \in \cI(G)}
\sigma_1\pi_{\bsigma|\bc} > \pi_{\boldsymbol{0}|\bc} $ when $c_1 >
\sum_{i=2}^n c_i.$ Since $k(x)$ is non-negative and $r_i k(x)$ is
strict increasing for all link $i,$ $r_i>0.$ Thus, when we devide both
sides of \eqref{eq:lin} by $r_i,$
\begin{equation}\label{eq:kproperty}
 k\left( \sum\limits_{i=2}^{n}c_i+ \frac{1}{m} \right) >
\sum\limits_{i=2}^n \frac{r_i}{r_1} \cdot k(c_i) > k\left(\sum\limits_{i=2}^{n}c_i
  - \frac{1}{m} \right). \end{equation}
By choosing $x=c_2=\dots=c_n$ and taking $m \to \infty$ in \eqref{eq:kproperty}, 
it follows that $\lim_{m\to\infty} \sum\limits_{i=2}^n \frac{r_i}{r_1}$
exists,\footnote{Recall that $[r_i]$ is a function of $m$.} and for any $0<x<1/n$,
 $$ k\left( (n-1)x \right) =  \lim_{m\to\infty} \sum\limits_{i=2}^n \frac{r_i}{r_1}\cdot k(x),$$
where $ \lim_{m\to\infty} \sum\limits_{i=2}^n \frac{r_i}{r_1}>1$ since $k(x)$ is strictly increasing.
Hence, if we take $x\to 0$ in the above inequality, $k(0)=0$ follows.
Similarly, by choosing $x=c_2$, $c_3=\dots=c_n=1/m$ and taking $m \to \infty$ in \eqref{eq:kproperty},
it follows that that for any $0<x<1$,
$$k(x) = \lim_{m\to\infty} \frac{r_2}{r_1}\cdot k(x),$$
where we use $k(0)=0$ and $\limsup_{m\to\infty}\frac{r_i}{r_1} <\infty$ due to 
the existence of $ \lim_{m\to\infty} \sum\limits_{i=2}^n \frac{r_i}{r_1}$. Thus, $\lim_{m\to\infty} \frac{r_2}{r_1}=1$,
and more generally, $\lim_{m\to\infty} \frac{r_i}{r_1}=1$ using same arguments.
Furthermore, by choosing $x=c_2$, $y=c_3$, $c_4=\dots=c_n=1/m$,  
and taking $m \to \infty$ in \eqref{eq:kproperty},
we have that for any $0<x+y<1$,
$$k(x+y) = k(x)+k(y),$$
where we use $\lim_{m\to\infty} \frac{r_2}{r_1}=1$. This implies that $k(x)$
is a linear function (with $k(0)=0$), and hence the conclusion of Theorem \ref{thm:nonlinear} follows.
 


\newcommand{\Z}{\mathbb{Z}}
\newcommand{\hs}{\hat{s}}
\newcommand{\hlambda}{\hat{\lambda}}

\section{Dynamic Throughput Optimal A-CSMA}

In the previous section, it is shown that, for any feasible arrival
rate, there exists an EXP-A-CSMA algorithm stabilizing the arrivals. In
this section, we describe EXP-A-CSMA algorithms which dynamically update
its parameters so as to stabilize the network without knowledge of the
arrival statistics. More precisely, the CSMA scheduling algorithm uses
$f_i^{(t)}$ and $g_i^{(t)}$ to compute the value of parameters $R_i
(t)=f_i^{(t)}(c_i(t))$ and $S_i (t)=g_i^{(t)}(c_i(t))$ at time $t$,
respectively, and update them adaptively over time.  We present two
algorithms to decide $f_i^{(t)}$ and $g_i^{(t)}$. They are building upon
prior algorithms in conjunction with the
properties of EXP-A-CSMA established in the previous section, referred
to as  a rate-based (extension of \cite{JSSW10DRA}) and queue-based
algorithm (extension of \cite{Shin12}). 


\subsection{Rate-based Algorithm}\label{sec:rate-based}
The first algorithm, at each queue $i$,  updates $(f_i^{(t)},g_i^{(t)})$ at time 
instances $L(j), j \in \Z_+$ with $L(0)=0$.  Also, $(f_i^{(t)},g_i^{(t)})$ remains fixed 
between times $L(j)$ and $L(j+1)$ for all $j \in \Z_+$. To begin
with, the algorithm sets $f_i^{(0)}(x)=g_i^{(0)}(x)=1$ (i.e, $R_i(0) = S_i(0)=1$) for all $i$ and all $x\in[0,1]$. With an abuse of notation,
$f_i^{(j)}$ denotes the value of $f_i^{(t)}$
for all $t \in [L(j), L(j+1))$. Similarly, we use $g_i^{(j)}$.
Finally, define $T(j) = L(j+1) - L(j)$
for $j \geq 0$.   

Now we describe how to choose a varying update interval $T(j)$. 
We select 
$ T(j)=\exp\left(\sqrt{j}\right),$ for $j \geq 1,$ and choose  
a step-size $\alpha(j)$ of the algorithm  as 
$\alpha(j)=\frac1{j},$ for $j\geq 1.$ 
Given this, queue $i$ updates $f_i$ and $g_i$ as follows. 
Let $\hlambda_i(j), \hs_i(j)$ be empirical arrival and service
observed at queue $i$ in $[L(j), L(j+1))$, \ie,
\begin{align*} 
&\hlambda_i(j) = \frac{1}{T(j)} A_i(L(j), L(j+1)) 
\qquad\text{and}\cr &\hs_i(t) = \frac{1}{T(j)} \left[ \int_{L(j)}^{L(j+1)} \sigma_i(t)c_i(t)\, dt \right].\end{align*}
Then, the update rule is defined by, for $x \in [0,1]$
\begin{eqnarray}
g_i^{(j+1)}(x) &=& \frac{j+2}{j+1} \cdot g_{i}^{(j)}(x)\cdot \exp\left(x\cdot \alpha(j)\cdot(\hs_{i}(j)-\hlambda_{i}(j))\right)\cr
f_i^{(j+1)}(x) &= &j+2 ,
\label{equpdate1}
\end{eqnarray}
with initial condition $f_i^{(0)}(x)=g_i^{(0)}(x)=1$. It is easy to check
that the A-CSMA algorithm with functions $[f_i^{(j)}]$ and $[g_i^{(j)}]$
lies in EXP-A-CSMA:
$$\log \frac{f_i^{(j)}(x)}{g_i^{(j)}(x)}~=~r_i(j)\cdot x,$$
where $r_i(0)=0$ and 
$$r_i(j+1)~=~r_i(j)+\alpha(j)\cdot(\hlambda_{i}(j)-\hs_{i}(j)).$$

Note that, under this update rule, the algorithm at each queue $i$
uses only its local history. Despite this, 
we establish that this algorithm is rate-stable, as formally stated as follows:
\begin{theorem}\label{thm:timevaryingcsma}
	For any given graph $G$, channel transition-rate $\bgamma$ and 
	$\blambda\in  \bLambda^o(\bgamma,G)$,
	the A-CSMA algorithm with updating functions
as per \eqref{equpdate1} 
is rate-stable.
\end{theorem}
The proof is presented in Appendix. 

\subsection{Queue-based Algorithm}\label{sec:queue-based}

The second algorithm chooses $(f_i^{(t)}, g_i^{(t)})$ as a simple function of queue-sizes as follows.
\begin{align}
&	f_i^{(t)}(x)=\Big( g_i^{(t)}(x) \Big)^2
	\quad\mbox{and}\quad \cr 
&g_i^{(t)}(x)=\exp\left(x\cdot \max\left\{w(Q_i(\lfloor t\rfloor)),\sqrt{w(Q_{\max}(\lfloor t\rfloor))}\right\}\right),\cr
\label{equpdate2} 
\end{align}
where $Q_{\max}(\lfloor t\rfloor)=\max_j Q_{j}(\lfloor t\rfloor)$ and $w(x)=\log\log(x+e)$.
One can interpret this as an EXP-A-CSMA algorithm since
$$\log \frac{f_i^{(t)}(x)}{g_i^{(t)}(x)}~=~r_i(t)\cdot x,$$
where $r_i(t) = \max\left\{w(Q_i(\lfloor t\rfloor)),\sqrt{w(Q_{\max}(\lfloor t\rfloor))}\right\}$.
The global information of 
$Q_{\max}(\lfloor t\rfloor)$ can be replaced by its approximate estimation that
can computed through a very simple distributed algorithm (with message-passing)
in \cite{SRS09} or a learning mechanism (without message-passing) in \cite{shah2011medium}. This does not
alter the rate-stability of the algorithm that is stated in the
following theorem, whose proof is presented in Appendix. 

\begin{theorem}\label{thm:queuecsma}
	For any given graph $G$, channel transition-rate $\bgamma$ and 
	$\blambda\in  \bLambda^o(\bgamma,G)$,
	the A-CSMA algorithm with functions
as per \eqref{equpdate2} 
is rate-stable.
\end{theorem}

\section{Achievable Rate Region of A-CSMA with Limited Backoff Rate}
\label{sec:fast}

In practice, it might be hard to have arbitrary large backoff rate
because of physical constraints. From this motivation, in this
section, we investigate the achievable rate region of A-CSMA
algorithms with limited backoff rate. Note that, in the proof of
Theorem~\ref{thm3}, we choose the backoff rates $[f_i]$ to be
proportional to the channel varying speed. Thus, when the backoff rate
is limited and the channel varying speed grows up, we cannot guarantee
the optimality of EXP-A-CSMA.  The main result of this section is
that, even with highly limited backoff rate, say at most $\delta>0$,
EXP-A-CSMA is guaranteed to have at least $\alpha$-throughput, where
$\alpha$ is {\em independent} of the channel varying speed and the
maximum backoff rate $\delta$. More formally, we obtain the following
result.

\begin{theorem}\label{thm:achi-rate-regi}
	For any $\phi>0$, interference graph $G$, channel transition-rate $\bgamma$
	and arrival rate $\blambda \in \alpha\bLambda^o$,
	there exists a rate-stable EXP-A-CSMA algorithm with functions $[f_i]$ and $[g_i]$ such that
	$$\max_{i\in V, x\in[0,1]} f_i(x)\leq \phi,$$
	where 
	\begin{equation}\label{eq:ucsma}
	\alpha = \max\left\{\min\limits_{i \in V} \sum_{\bc \in \sch^n}c_i \pi_{\bc},\frac1{\chi(G)}\right\}.
	\end{equation}
	In above, $\chi(G)$ is the chromatic number of $G$.
\end{theorem}

Theorem \ref{thm:achi-rate-regi} implies that even with arbitrary
small backoff rates, A-CSMA is guaranteed to achieve at least a
partial fraction of capacity region. For example, for bipartite
interference graph, at least 50\%-throughput is achieved. 
The proof strategy is as follow: {\em step 1)} we find the achievable rate
region of U-CSMA and {\em
  step 2)} we show that, for any U-CSMA parameters, there exists a
EXP-A-CSMA algorithm satisfying the backoff constraint and achieving 
$\varepsilon$ close departure rate with the U-CSMA algorithm (we
formally state this in Corollary~\ref{cor:UACSMA}).
\begin{corollary}\label{cor:UACSMA}
For any 
$\phi>0$, interference graph $G$, channel transition-rate
$\bgamma,$ and U-CSMA parameters, there
exists a EXP-A-CSMA algorithm with functions $[f_i]$ and
$[g_i]$ such that $\max_{i\in V, x\in[0,1]} f_i(x)\leq \phi$ and
$$\lim\sup_{t\rightarrow \infty}\left|1-\frac{\widehat{D}^{A}_i
    (t)}{\widehat{D}^{U}_i (t)}\right|~ <~ \varepsilon,~\mbox{for
  all}~i \in V,$$
where $\widehat{D}^{A}_i$ and $\widehat{D}^{U}_i$ denote the
cumulative potential departure processes of the EXP-A-CSMA and the
U-CSMA, respectively.
\end{corollary}

\subsection{Proof of Theorem~\ref{thm:achi-rate-regi}}\label{sec:pfthm2}

The main strategy for the proof of Theorem \ref{thm:achi-rate-regi}
is that we study U-CSMA (channel-unaware CSMA) to achieve the performance guarantee of A-CSMA.
We start by stating the following key lemmas about U-CSMA.
\begin{lemma}\label{lem1}
	Let $P_I(G)$ be the independent-set polytope, \ie,
	\begin{equation}
	P_I(G)=\left\{\bm{x}\in[0,1]^n: \bm{x} = \sum_{\brho \in \cI(G)}\alpha_\brho \brho, \sum_{\brho \in \cI(G)} \alpha_{\brho}=1,\alpha_{\boldsymbol 0}>0   \right\}.
	\end{equation}
Then, for $\blambda\in P_I(G)$, there exists a U-CSMA algorithm with parameters $\boR=[R_i]$ and $\boS=[S_i]$ such that
	$$ \lim_{t\to\infty}\mathbb E[\bsigma(t)]>\blambda.$$
\end{lemma}
\begin{proof}
	The proof of Lemma 8 in \cite{JSSW10DRA} goes through for the proof of Lemma \ref{lem1} in an identical manner. We omit further details.
\end{proof}

\begin{lemma}\label{thm2}
  For any $\phi>0$, interference graph $G$, channel transition-rate $\bgamma$ and arrival rate $\blambda \in \alpha\bLambda^o$, 
there exists a rate-stable U-CSMA algorithm with
parameters $\boR=[R_i]$ and $\boS=[S_i]$ such that $$\max_i R_i\leq \phi\quad\mbox{and}\quad\max_i S_i\leq \phi,$$
where $\alpha$ is defined in \eqref{eq:ucsma}.
\end{lemma}

\begin{proof}
	It suffices to show that there exists a U-CSMA algorithm
        stabilizing any arrival rate $\blambda$ such that $$\blambda \in \frac1{\chi(G)}\cdot\bLambda^o
	\quad\mbox{or}\quad\blambda \in \min\limits_{i \in V} \sum_{\bc \in \sch^n}c_i \pi_{\bc}\cdot \bLambda^o.$$

First, consider
$\blambda \in \frac1{\chi(G)}\cdot\bLambda^o$. 
From Lemma \ref{lem2} and the ergodicity of Markov process $\{(\bsigma(t)\}$ and $\{\bc(t)\}$ under U-CSMA,
it suffices to prove that
there exists a U-CSMA algorithm satisfying
	$$ \lim_{t\to\infty}\mathbb E[\sigma_i(t)c_i(t)]> \lambda_i\qquad\mbox{for all}~i\in V.$$
	Since ${\chi(G)}\cdot \lambda_i<\lim_{t\to\infty}\mathbb E[c_i(t)]= \sum_{\bc \in \sch^n}c_i \pi_{\bc}$ (otherwise, ${\chi(G)}\blambda \notin \bLambda^o$),
	it is enough to prove that for an appropriately defined $\delta>0$,
	$$ \lim_{t\to\infty}\mathbb E[\sigma_i(t)]>\frac1{\chi(G)}-\delta\qquad\mbox{for all}~i\in V.$$
	There exists a U-CSMA algorithm with parameter $\boR=[R_i]$ and $\boS=[S_i]$ satisfying
	the above inequality from Lemma \ref{lem1} and
	$\left[\frac1{\chi(G)}-\delta\right]\in P_I(G)$. Furthermore, we can make $R_i$ and $S_i$ arbitrarily small since
	$ \lim_{t\to\infty}\mathbb E[\sigma_i(t)]$ under U-CSMA is invariant as long as ratios $R_i/S_i$ remain same.

	Now the second case $\blambda \in \min\limits_{i \in V} \sum_{\bc \in \sch^n}c_i \pi_{\bc}\cdot \bLambda^o$ can be proved
	in an similar manner, where we have to prove that 
	there exists a U-CSMA algorithm satisfying
		$$ \lim_{t\to\infty}\mathbb E[\sigma_i(t)]>\rho_i\qquad\mbox{for all}~i\in V,$$
		where we define $\brho=[\rho_i]$ as
		$$\brho = \frac1{\min\limits_{i \in V} \sum_{\bc \in \sch^n}c_i \pi_{\bc}}\cdot \blambda \in \bLambda^o\subset P_I(G).$$
		This follows from Lemma \ref{lem1} and $\brho\in P_I(G)$.
		This completes the proof of Lemma
                \ref{thm2}.\end{proof}

\begin{figure*}[t!]
  \begin{center}
    \subfigure[\small 5-link complete graph]{
      \includegraphics*[width = 0.32\columnwidth]{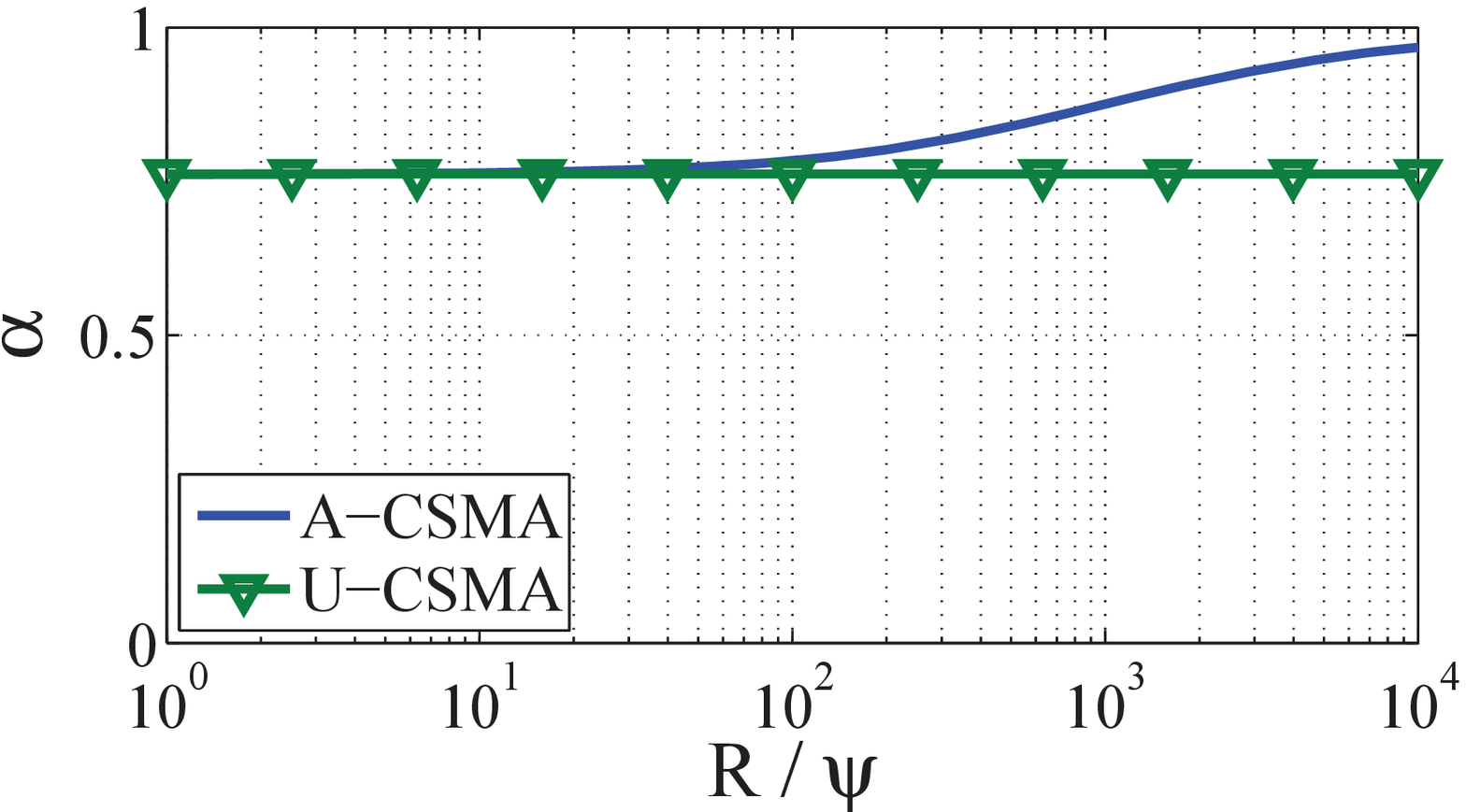}\label{fig:comp1} }
    \hspace{0.5cm}
    \subfigure[\small Random topology]{
      \includegraphics*[width = 0.2\columnwidth]{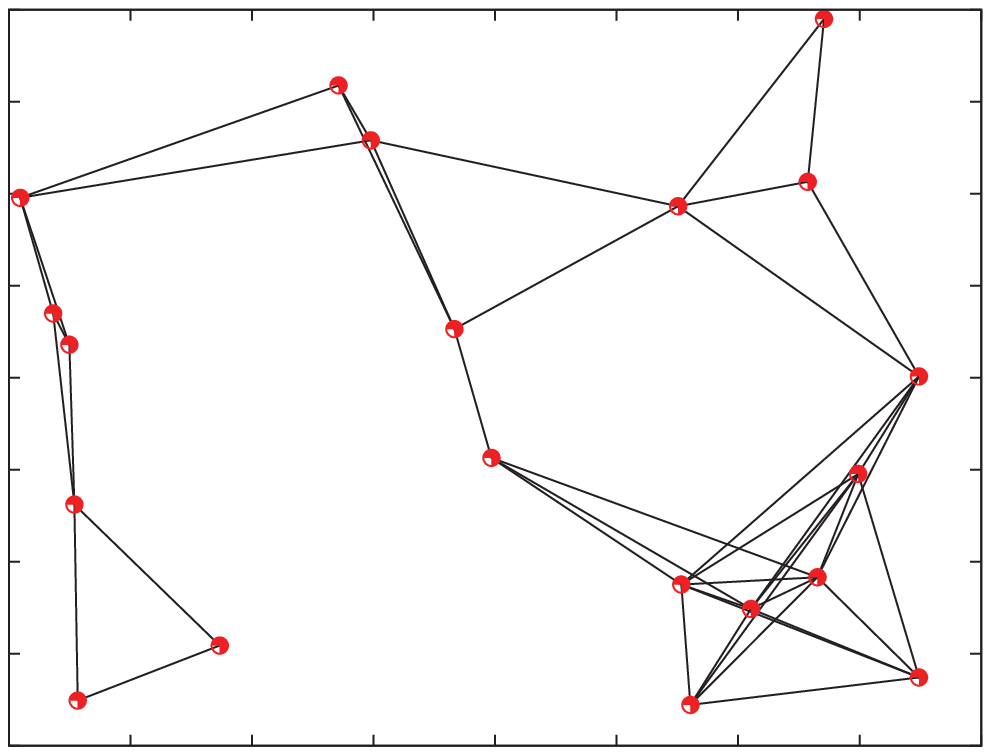}\label{fig:rand_topo}}
    \hspace{0.5cm}
    \subfigure[\small Random: A-CSMA vs U-CSMA and Uniqueness]{
      \includegraphics*[width = 0.32\columnwidth]{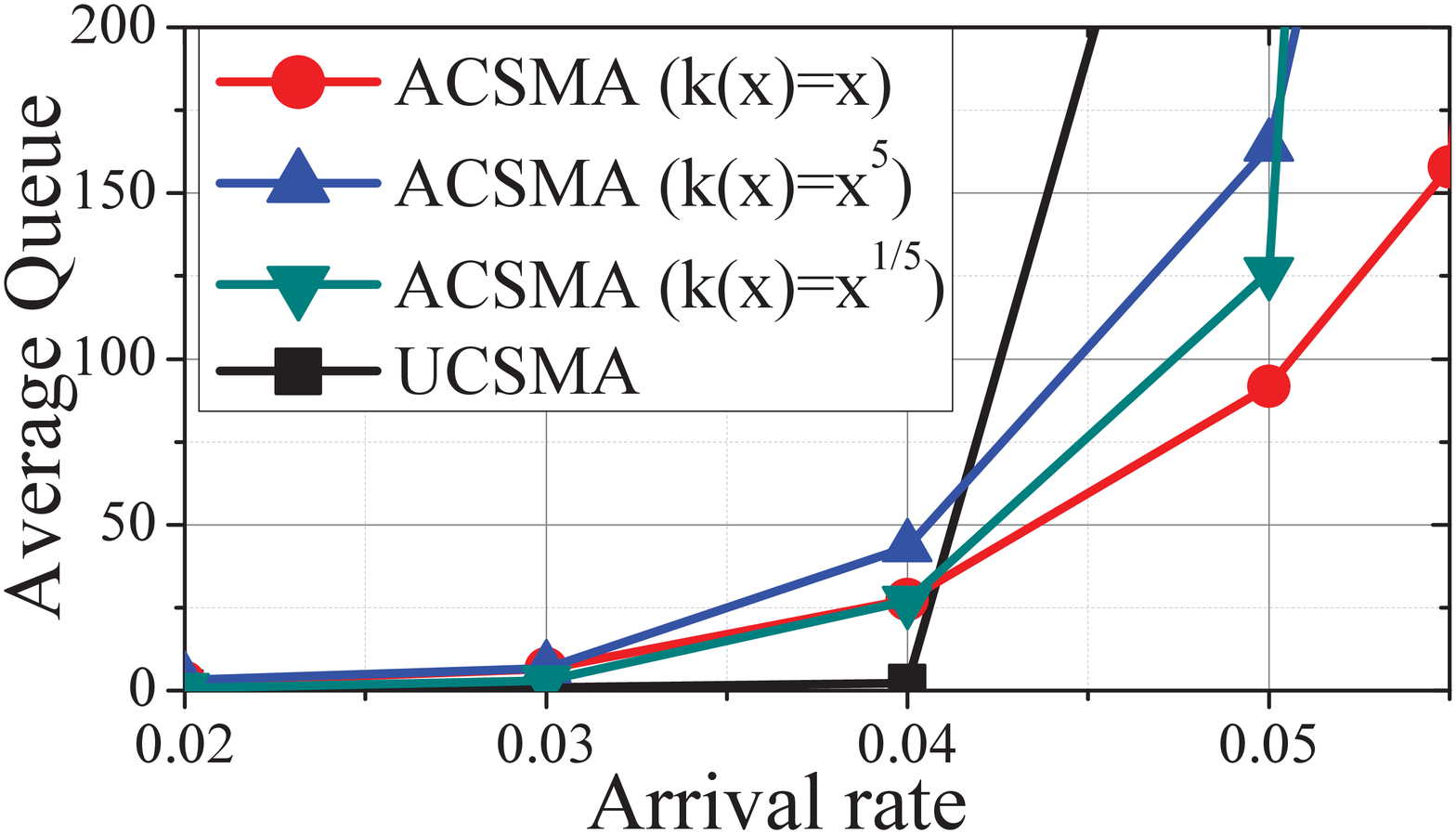}\label{fig:comp2} }
  \vspace{-0.2cm}
    \caption{Numerical Results}
    \label{fig:numerical}
  \end{center}
  \vspace{-0.2cm}
\end{figure*}

Lemma \ref{thm2} implies that
for any arrival rate
$\blambda=[\lambda_i]\in \alpha \bLambda^o,$ there exist $\varepsilon>0$ and
a rate-stable U-CSMA algorithm with
arbitrary small parameters $[R_i]$ and $[S_i]$, which stabilize arrival rate $(1+\varepsilon)\blambda$, \ie,
$$(1+\varepsilon)\lambda_i
\leq\sum_{\bc \in \sch^n}c_i\pi_{\bc}\sum_{\bsigma \in \cI(G): \sigma_i = 1}\pi^*_{\bsigma},$$
where $[\pi^*_{\bsigma}]$ is the stationary distribution of Markov
process $\{ \bsigma(t)\}$ induced by the U-CSMA algorithm. In particular, given $\phi>0$,
one can assume 
$\max_i R_i\leq \phi$.
For the choice of $[R_i]$ and $[S_i]$, 
we consider an EXP-A-CSMA algorithm with functions $$f_i(x) = R_i\quad\mbox{and}\quad g_i(x) =  R_i \exp
(-r_i x),$$ where we choose $r_i$ to satisfy
$$S_i = \sum\limits_{\bc \in \sch^n}
\pi_{\bc}R_i \exp(-r_i \cdot c_i).$$
Note that $r_i$ satisfying the above equality always exists for given $S_i$,
and
$$\max_{i\in V, x\in[0,1]} f_i(x)=\max_i R_i\leq \phi.$$
Furthermore, one can observe that the maximum value of $f_i(x)$ and $g_i(x)$ for $x\in [0,1]$
can be made arbitrarily small due to arbitrarily small $R_i,S_i$.
Using this observation and the Markov chain tree theorem (as we did 
for the proof of Lemma~\ref{lem:thrwithin}), one can show that
$$\max_{(\bsigma,\bc)\in \mathcal I(G)\times\sch^n}\left|1-\frac{\pi_{\bsigma,\bc}}{\pi_{\bc}\pi^*_{\bsigma}}\right|~ <~ \varepsilon,$$
where $[\pi_{\bsigma,\bc}]$ denotes the stationary distribution of
Markov process $\{(\bsigma(t),\bc(t)) \}$ by the EXP-A-CSMA algorithm.  Therefore, it follows that
\begin{eqnarray*}
	\lambda_{i} &\leq&\left(1-\frac{\varepsilon}{1+\varepsilon}\right)\sum_{\bc \in \sch^n}c_i\pi_{\bc}\sum_{\bsigma \in \cI(G): \sigma_i = 1}\pi^*_{\bsigma}\\
	&<&~ \sum_{\bc\in\sch^n}\sum_{\bsigma \in \cI(G): \sigma_i = 1}c_i\pi_{\bsigma,\bc}\\
	&=&~\lim_{t\to\infty} \frac1t \widehat D_{i}(t),
\end{eqnarray*}
where the last inequality is from the ergodicity of Markov process
$\{(\bsigma(t),\bc(t))\}$. Due to Lemma \ref{lem2}, this means that 
the EXP-A-CSMA algorithm is rate-stable for the arrival rate $\blambda$.
This completes the proof of Theorem \ref{thm:achi-rate-regi}.

\section{Numerical Results}
\label{sec:num}

In this section, we provide simple numerical results that demonstrate
our analytical findings. 

\smallskip
\noindent{\bf \em Complete interference graph.}
We first consider a 5-link complete interference graph, \ie,
all  5 links interfere with each other.  All
queues are homogeneous in terms of time-varying channels, where we
assume that the channel space is simply $\{0.5,
1\}$ and the transition-rate $\gamma=\ratech^{0.5 \to 1} =
\ratech^{1 \to 0.5}.$ We compare A-CSMA and U-CSMA, with the following
functions: 
\begin{eqnarray*}
  \text{A-CSMA:} && f_i(x) = R, \quad g_i(x) = R\cdot10^{-4x}\cr
  \text{U-CSMA:} &&f_i(x) = R, \quad g_i(x) = R\cdot10^{-4},
\end{eqnarray*}
so that $\log(f_i/g_i) = 4 x$ for A-CSMA and $4$ for U-CSMA, respectively. 
Throughputs of A-CSMA and U-CSMA are evaluated
by estimating the average rate in the potential departure process, \ie, $\lim_{t\to\infty} \frac1t\widehat{\bD}(t)$.
Figure~\ref{fig:comp1} shows the results, where in $x$-axis, we vary the
ratio of backoff rate $R$ to channel varying speed $\psi$ (determined by
$\gamma$) and $y$-axis represents the fraction of achievable rate region
$\alpha$ (note that in a complete interference graph, the rate region is
symmetric). We see that (i) by reflecting the channel capacity in the
CSMA parameters as an exponential function, A-CSMA has $\alpha$-throughput
where $\alpha$ approaches
$100\%$ (this can be explained by Theorem~\ref{thm3}), and (ii) U-CSMA has $76\%$-throughput.
Note that $\alpha \geq 76\%$
even with limited backoff rates (\ie, small $R/\psi$), and this
matches Corollary~\ref{cor:UACSMA} which states that
A-CSMA's throughput is at least U-CSMA's throughput.


\smallskip
\noindent{\bf \em Random topology.}
We now study {\em dynamic} A-CSMA and U-CSMA for a random topology
by uniformly locating 20 nodes in a square area and a link between two
nodes are established by a given transmission range, as
depicted in Figure~\ref{fig:rand_topo}. 
To model interference, we assume the two-hop interference model (\ie, any two links within two
hops interfere) as in 802.11. Here, each link has independent and
identical channels, where $\sch=\{\frac{u}{10}: 1\leq u \leq 10\}.$ For
all link $i,$ $\ratech^{u/10 \to (u+1)/10} = \ratech^{u/10 \to
  (u-1)/10}= 0.01,$ and 0 otherwise.

In Figure~\ref{fig:comp2}, we increase the arrival rates homogeneously
for all links, and plot the average queue lengths to see which arrival
rates makes the system stable or unstable across the tested
algorithms. The average queue length blows up when the algorithm
cannot stabilize the given arrival rate. We test {\em dynamic} A-CSMA
and U-CSMA algorithms: the queue-based A-CSMA($x$), A-CSMA($x^5$),
A-CSMA($x^{1/5}$), and U-CSMA, where for given function $k(x)$,
A-CSMA($k(x)$) denotes the A-CSMA algorithms satisfying
\eqref{eq:kcsma}. Note that if $k(x)=1$, A-CSMA($k(x)$) is equal to
U-CSMA. The functions $[f_i]$ and $[g_i]$ are defined as stated in
Section~\ref{sec:queue-based} except the channel adaptation function
$k(\cdot)$. Figure~\ref{fig:comp2} shows that (a) A-CSMA($x$)
stabilize more arrival rates than A-CSMA($x^5$) and A-CSMA($x^{1/5}$),
which coincides with our uniqueness result (see
Theorem~\ref{thm:nonlinear}) and (b) dynamic A-CSMA algorithms
outperforms dynamic U-CSMA when the arrival rate is larger than 0.04,
which means that the achievable rate region of A-CSMA includes the
achievable rate region of U-CSMA. In low arrival rate region, U-CSMA
could be better than A-CSMA in view of delay because the transmission
intensity of U-CSMA is always high when the queue is large, while,
under A-CSMA algorithm, each link waits until its channel condition
being good although the queue is large.

  


\section{Conclusion}
Recently, it is shown that CSMA algorithms can achieve throughput (or
utility) optimality where `static' channel is assumed. However, in
practice, the channel capacity of each link has variation. To our best
knowledge, this work is the first study on the throughput optimality
with time-varying channels. To this end, we propose A-CSMA wich
adaptively acts on the channel variation. First we show that the
achievable rate region of A-CSMA contains all of the capacity of the
network.  From the result, in this work, we design throughput optimal
A-CSMA algorithms which can stabilize any arrival rate in the
capacity. We also consider more practical scenario of limited backoff
rate. According to our results, with any backoff rate limitation, the
achievable rate region of A-CSMA contains the achievable rate region
of channel unaware CSMA (U-CSMA) which does not adapt to the channel
variations.

\bibliographystyle{abbrv}
\bibliography{reference}

\appendix
\noindent{\bf \em Proof of Theorem~\ref{thm:timevaryingcsma}.}
Despite the fact that the dynamics of channel and CSMA Markov chain are
coupled in a complex manner, we found that the proof of Theorem
\ref{thm:timevaryingcsma} largely shares with that of Theorem 1 in
\cite{JSSW10DRA} in conjunction with Lemma \ref{lem:thropt} and Lemma
\ref{lem:thrwithin}. Thus, we provide the proof sketch with focus on the
key step, as described by:
$\lim_{j\to \infty}\br(j)=\br^*,$ with probability 1,
where $\br(j)=\left[\log \frac{f_i^{(j)}(1)}{g_i^{(j)}(1)}\right]$ and
$\br^*=\br^*(\blambda)$ is the unique maximizer of the following
function $F$:
\begin{equation*}
  F(\bm{r}) = \bm{\lambda}\cdot\bm{r} - \sum_{\bc \in \sch^n
  }\pi_{\bc}\log \left( \sum_{\bsigma \in \cI(G)}
    \exp\left(\sum_{i}\sigma_i c_i r_i\right) \right),
\end{equation*}
where the existence and uniqueness of the maximum point is guaranteed
by Lemma~\ref{lem:thropt}.
The maximum point $\br^*$ can be derived by establishing 
that the updating rule \eqref{equpdate1} is the (stochastic) gradient algorithm maximizing $F$, \ie,
\begin{eqnarray}
  \frac{\partial F}{\partial r_i}(\bm{r}(j))&=&\lambda_i -
  \sum_{\bc \in \sch^n}c_i \pi_{\bc}\sum_{\bsigma \in \cI(G): \sigma_i = 1}\pi_{\bsigma|\bc}(j)
  \label{eq1:pfthmtimevaryingcsma}\\
  &\stackrel{j\to\infty}{=}&\lambda_i-\sum_{\bc \in \sch^n}\sum_{\bsigma \in \cI(G): \sigma_i = 1}c_i\pi_{\bsigma,\bc}(j)\label{eq2:pfthmtimevaryingcsma}\\
  &\stackrel{j\to\infty}{=}& \hlambda_{i}(j)-\hs_{i}(j),\label{eq3:pfthmtimevaryingcsma}\end{eqnarray}
where $[\pi_{\bsigma,\bc}(j)]$ and $[\pi_{\bsigma|\bc}(j)]$ are
the stationary distributions of Markov processes $\{(\bsigma(t),\bc(t))\}$ 
and $\{\bsigma(t),\bc\}$ induced by the A-CSMA algorithm with functions $[f_i^{(j)}]$ and $[g_i^{(j)}]$, respectively.
(\ref{eq1:pfthmtimevaryingcsma}) and (\ref{eq3:pfthmtimevaryingcsma})
can be shown using a similar strategy to those in \cite{JSSW10DRA}. The second equality (\ref{eq2:pfthmtimevaryingcsma}) follows from Lemma \ref{lem:thrwithin}
and for all channel state $\chs_u,$ $f_i^{(j)}(\chs_u),g_i^{(j)}(\chs_u)\to\infty$ as $j\to\infty$. We omit
further necessary details.

\smallskip
\noindent{\bf \em Proof of Theorem~\ref{thm:queuecsma}.}
Similar to the proof of Theorem~\ref{thm:timevaryingcsma}, for the
queue-based algorithm, we sketch the proof by focusing on the key
difference from the proof of the prior algorithm \cite{Shin12}. 
It suffices to show that the underlying Markov process is positive
recurrent, which implies the rate-stability in this paper. The key
difference part is to prove the following (which is a generalized
version of Lemma 7 in \cite{Shin12}):  
for given $\varepsilon>0$ and $[Q_i]\in \mathbb R_+^n$ with large enough $Q_{\max}=\max_i Q_i$ (depending on $\varepsilon$ and $n$),
we have to show that
\sqeq
\begin{equation}
 \mathbb E\left[\sum_i w(Q_i)x_iy_i \right] ~~\geq~~ 
  (1-\varepsilon)\cdot\mathbb E\left[\max_{\bz\in\cI(G)}\sum_i w(Q_i)z_i y_i \right].\label{eq1:pfthmqueuecsma}
\end{equation}\unsqeq
In above, $\bx=[x_i]\in\cI(G)$, $\by=[y_i]\in\sch^n$ and $(\bx,\by)$ is distributed as per $[\pi_{\bsigma,\bc}]$ that is
the stationary distribution of Markov
process $\{(\bsigma(t),\bc(t))\}$ induced by the A-CSMA algorithm with
parameters $[f_i]=[g_i^2]$ and
$$ [g_i (c_i)]=\left[\exp\big(c_i \cdot
  \max\{w(Q_i),\sqrt{w(Q_{\max})}\}\big)\right]. $$
To show \eqref{eq1:pfthmqueuecsma}, it suffices to prove that (with probability 1)
\sqeq
\begin{equation}
  \mathbb E\left[\sum_i w(Q_i)x_iy_i-(1-\varepsilon)\cdot \max_{\bz\in\cI(G)}\sum_i w(Q_i)z_i y_i~\Big|~\by \right]~\geq~0.\label{eq2:pfthmqueuecsma}
\end{equation}\unsqeq
To this end, Lemma \ref{lem:thrwithin} 
implies that the distribution of $\bx$ given $\by$ is close to $[\pi_{\bsigma | \by}]$ that is 
the stationary distribution of Markov
process $\{(\bsigma(t),\by)\}$ induced by the A-CSMA algorithm with
parameters $[f_i]=[g_i^2]$ and
$ [g_i (c_i)]=\left[\exp\big(c_i \cdot
  \max\{w(Q_i),\sqrt{w(Q_{\max})}\}\big)\right]. $
 This is because 
$f_i (\chs_u)$ and $g_i (\chs_u)$ are large enough for all
channel state $h_u$ due to large enough $Q_{\max}$. Namely, for $\bsigma=[\sigma_i]\in \cI(G)$, it follows that
\begin{align}
  &\Pr[\bx=\bsigma~|~\by]\approx\frac1 Z\exp\left(\sum_i \sigma_i\cdot \log \frac{f_i(y_i)}{g_i(y_i)}\right)
  ,\label{eq3:pfthmqueuecsma}
\end{align} 
where $Z=\sum_{\bsigma\in\mathcal I^{\by}(G)}\exp\left(\sum_i \sigma_i\cdot \log \frac{f_i(y_i)}{g_i(y_i)}\right)$ is the normalizing factor.
Now one can prove \eqref{eq2:pfthmqueuecsma} using \eqref{eq3:pfthmqueuecsma} following the same strategy with that in the proof of
Lemma 7 in \cite{Shin12}. We omit further necessary details.


\end{document}